\documentclass[journal]{IEEEtran}
\usepackage[english]{babel}  
\usepackage{amsmath,amssymb,paralist,subfigure,graphicx,stmaryrd,mathrsfs,url,algorithm2e,soul}
\usepackage{xcolor}
\usepackage{cuted,mathtools,lipsum,cite}

\newtheorem{lem}{Lemma}
\newtheorem{ass}{Assumption}
\newtheorem{thm}{Theorem}
\newtheorem{defn}{Definition}
\newtheorem{rem}{Remark}

\def\mb{\mathbf}

\def\mc{\mathcal}

\DeclareMathOperator*{\argmin}{argmin}
\DeclareMathOperator*{\argmax}{argmax}

\newcommand{\ab}[1]{\textcolor{blue}{{#1}}}
\begin{document}
\title{\huge \bf 1st-Order Dynamics on Nonlinear Agents for Resource Allocation over Uniformly-Connected Networks}
\author{Mohammadreza Doostmohammadian, Alireza Aghasi, Maria Vrakopoulou,  
Themistoklis Charalambous
	
\thanks{M. Doostmohammadian and T. Charalambous are with the School of Electrical Engineering at Aalto University, Finland, email: \texttt{firstname.lastname@aalto.fi}. M. Doostmohammadian is also with the Faculty of Mechanical Engineering at Semnan University,  Iran, email: \texttt{doost@semnan.ac.ir}. T. Charalambous is also with the University of Cyprus, Nicosia, Cyprus.  
Alireza Aghasi is with Robinson College of Business, Georgia State University, USA, email: \texttt{aaghasi@gsu.edu}. Maria Vrakopoulou is with the University of Melbourne, Australia, email: \texttt{maria.vrakopoulou@unimelb.edu.au}. 
}}
\maketitle
\thispagestyle{empty}
\begin{abstract}
	A general nonlinear $1$st-order consensus-based solution for distributed constrained convex optimization  is proposed with network resource allocation applications. The solution is used to optimize  continuously-differentiable \textit{strictly convex} cost functions over weakly-connected undirected  networks, while it is \textit{anytime feasible}  and  models various nonlinearities to account for imperfections and constraints on the (physical model of) agents in terms of limited actuation capabilities, e.g.,  quantization and saturation. \ab{Due to such inherent nonlinearities, the existing linear solutions considering ideal agent models may not necessarily converge with guaranteed optimality and anytime feasibility.  Some applications also impose specific nonlinearities, e.g., convergence in fixed/finite-time or sign-based robust disturbance-tolerant dynamics. Our proposed distributed protocol generalizes such nonlinear models. Putting convex set analysis together with nonsmooth Lyapunov analysis, we prove convergence, (i) regardless of the particular type of nonlinearity, and (ii) with weak network-connectivity requirements (\textit{uniform-connectivity}).} 
\end{abstract}
\vspace{-0.4cm}
\section{Introduction} \label{sec_intro}
\IEEEPARstart{C}{onsensus}
has been infiltrated into  control and machine learning, e.g., in \textit{distributed} optimization \cite{xin2019distributed}, estimation \cite{tcns2020,ISJ_cyber}, and  resource allocation \cite{gharesifard2013distributed}. \textit{Network resource allocation} is the problem of allocating constant amount of resources among agents to minimize the cost, with  application to several fields, such as, the distributed Economic Dispatch Problem (EDP)
\cite{cherukuri2015tcns,yu2017distributed,molzahn2017survey,yang2013consensus,chen2018fixed,chen2016distributed,li2017distributed,yi2016initialization}, distributed coverage control
\cite{higuera2012distributing,MSC09},  congestion control \cite{srikant2004mathematics}, and distributed load balancing  \cite{mach2017mobile}.
Such  problems are subject to inherent \textit{physical constraints} on the agents,  
leading to nonlinear dynamics with respect to actuation and affecting the stability. This work formulates a general solution considering such nonlinear agents to solve distributed allocation. \ab{Another example is the Automatic Generation Control (AGC) in electric power systems  \cite{kundur, hiskens_agc} which regulates the generators' output power compensating for any generation-load mismatch in the system. The AGC generators' deviations are subject to limits based on the available power reserves and also on Ramp Rate Limits (RRLs) (or rate saturation), i.e, the speed their produced power can increase or decrease is constrained. 
Under such nonlinear constraints, a linear (ideal) model for generators as given by \cite{cherukuri2015tcns,gharesifard2013distributed} may not remain feasible or result in a sub-optimal solution.}
 
\textit{Related literature:} The literature spans from preliminary linear  \cite{gharesifard2013distributed,cherukuri2015tcns,boyd2006optimal} and accelerated linear \cite{shames2011accelerated} solutions to more recent  sign-based consensus \cite{wang2020distributed}, Newton-based  \cite{martinez2021distributed}, derivative-free swarm-based  \cite{hui2014optimal}, predictive online saddle-point methods \cite{chen2017online}, $2$nd-order autonomous dynamics \cite{deng2019distributed,yu2017distributed,deng2017distributed,wang2018second}, distributed mechanism  over local message-passing networks \cite{heydaribeni2019distributed}, multi-objective  \cite{li2020distributed}, \ab{primal-dual \cite{turan2020resilient,nesterov2018dual,feijer2010stability}, Lagrangian-based  \cite{xu2017distributed,dominguez2012decentralized,kar2012distributed,li2017distributed}, and projected proximal sub-gradient algorithms \cite{iiduka2018distributed}, among others. These works cannot address different inherent physical nonlinearities on the agents' model, such as the RRL for distributed AGC, or some other designed nonlinear models intended for improving computation load and convergence rate, e.g., reaching fast convergence. In general, model nonlinearities such as limited computational capacities, constrained actuation, and model imperfections may significantly affect the convergence or degrade the resource allocation performance. For example, none of the mentioned references can address quantization, saturation, and sign-based actuation altogether, while ensuring feasibility at all times. In reality, under such model nonlinearities there is no guarantee that the existing solutions accurately follow the ideally-designed dynamics and preserve feasibility, optimality, or specified convergence rate.  Some existing Lagrangian-based methods \cite{dominguez2012decentralized,kar2012distributed} are not \textit{anytime} feasible, but reach feasibility upon the convergence \cite{cherukuri2015tcns}.  In a different line of research, \textit{inequality-constrained} problems are solved via primal-dual methods and Lagrangian relaxation  \cite{turan2020resilient,nesterov2018dual,feijer2010stability}. This differs from equality-constrained problems which are typically solved via Laplacian gradient methods. The latter is is used for the optimal resource allocation in EDP \cite{cherukuri2015tcns,yu2017distributed,molzahn2017survey}, but without addressing the RRL nonlinearity on the power \textit{rate}.
}

\textit{Main contributions:}
We propose a general $1$st-order  \textit{Laplacian-gradient dynamics} for distributed resource allocation. The proposed localized solution generalizes many nonlinear constraints on the agents including, but not limited to, (i) \textit{saturation} and (ii) \textit{quantization}. Further, some specific constraints (e.g., on the convergence or robustness), impose nonlinearities on the agents' dynamics. For example, it is practical in  applications  to design (iii) \textit{fixed-time} and \textit{finite-time} convergent solutions, and/or (iv) \textit{robust} protocols to impulsive noise and   uncertainties.  
Our proposed dynamics generalizes many \textit{symmetric sign-preserving} model nonlinearities.
\ab{We prove uniqueness, anytime feasibility, and convergence over generally sparse, time-varying, undirected (and not necessarily connected)  networks, referred to as \textit{uniform-connectivity}. The proofs are based on nonsmooth Lyapunov theory \cite{cortes2008discontinuous}, graph theory, and convex analysis,  \textit{irrespective of the type of nonlinearity}. This generalized 1st-order solution is more practical as it considers all possible  sign-preserving physical constraints on the agents dynamics, and further, can be extended to consider nonlinearities on the agents' communications \cite{taes2020finite,mrd_2020fast}.}
\vspace{-0.2cm}
\section{Problem Statement} \label{sec_prob}
The network resource allocation problem  is in the form\footnote{Note the subtle abuse of notation where the overall state $\mb{X}$ is represented in matrix form  to simplify the notation in proof analysis throughout the paper. },
\begin{align} \label{eq_dra}
\min_\mb{X}
 F(\mb{X},t) = \sum_{i=1}^{n} f_i(\mb{x}_i,t), ~
\text{s.t.} ~  \mb{X}\mb{a} = \mb{b}
\end{align}  
with $\mb{x}_i \in \mathbb{R}^d$,
$\mb{X} = [\mb{x}_1,\dots,\mb{x}_n] \in \mathbb{R}^{d \times n}$, vectors $\mb{a}=[{a}_1;\dots;{a}_n] \in \mathbb{R}^n$, and $\mb{b}=[{b}_1;\dots;{b}_d] \in \mathbb{R}^d$. \ab{The entries of $\mb{a}$ are assumed to  not be  very close to zero to avoid unbounded solutions. If ${a}_j =0$ for agent $j$, its state $\mb{x}_j$ is decoupled from the other agents, and  problem~\eqref{eq_dra} can be restated for $n-1$ coupled agents plus an unconstrained optimization on $f_j(\mb{x}_j,t)$. } $f_i(\mb{x}_i,t):\mathbb{R}^{d+1} \rightarrow \mathbb{R}$ in \eqref{eq_dra} denotes the local time-varying cost  at agent $i$ as $f_i(\mb{x}_i,t) = \widetilde{f}_i(\mb{x}_i) + \widehat{f}_i(t)$, with $\widehat{f}_i(t) \neq 0$ representing the time-varying cost. 
In some applications, the states are subject to the \textit{box constraints}, $\underline{\mb{m}} \leq \mb{x}_i \leq \overline{\mb{m}}$, denoting element-wise comparison. 
Using exact penalty functions, these constraints are added into the local objectives as $f_i^\epsilon (\mb{x}_i,t) = f_i(\mb{x}_i,t) + \epsilon h^\epsilon(\mb{x}_i - \overline{m}) + \epsilon h^\epsilon(\underline{m} - \mb{x}_i)$ with $h^\epsilon(u)=  \max \{u,\mb{0}\}$. The smooth equivalent substitutes are $\frac{1}{\mu} \log (1+\exp (\mu u))$ \cite{csl2021}, \textit{quadratic} penalty $(\max \{u,\mb{0}\})^2$ 
(or $\theta$-logarithmic barrier \cite{li2017distributed})
with the gap inversely scaling with $\epsilon$.

\begin{ass} \label{ass_strict}
	The (time-independent part of) local functions, $\widetilde{f}_i(\mb{x}_i):\mathbb{R}^{d} \rightarrow \mathbb{R}$, are strictly convex and differentiable. 
\end{ass}

This assumption ensures unique optimizer (see Lemma~\ref{lem_unique_feasible}) and existence of function gradient. 
This paper aims to design a \textit{localized} general nonlinear dynamic to solve \eqref{eq_dra} based on partial information at agents  over a network.
\vspace{-0.4cm}
\section{Definitions and Auxiliary Results }
\label{sec_def}
\subsection{ Graph Theory and Nonsmooth Analysis}
The multi-agent network is modeled  as a time-varying undirected graph $\mc{G}(t)=\{\mc{V},\mc{E}(t)\}$ with links $\mc{E}(t)$ and nodes $\mc{V}=\{1,\dots,n\}$. $(i,j) \in \mc{E}(t)$ denotes a link from agent $i$ to  $j$, and the set $\mc{N}_i(t)=\{j|(j,i)\in \mc{E}(t)\}$ represents the direct neighbors of agent $i$ over  $\mc{G}(t)$. Every  link $(i,j) \in \mc{E}(t)$ is assigned with a positive weight $W_{ij}>0$, in the associated weight matrix $W(t)=[W_{ij}(t)] \in \mathbb{R}^{n \times n}_{\geq0}$ of $\mc{G}(t)$. In $\mc{G}(t)$ define a \textit{spanning tree} as a subset of links in which there exists only one path between every two nodes (for all $n$ nodes). 
\begin{ass} \label{ass_G}
	The following assumptions hold on $\mc{G}(t)$:
	\begin{itemize}
		\item The network $\mc{G}(t)$ is undirected.  This implies a symmetric associated weight matrix $W(t)$, i.e., $W_{ij}(t)=W_{ji}(t)\geq0$ for $i,j \in \{1,\dots,n\}$ at all time $t\geq0$, which is not necessarily row, column, or doubly stochastic.
		\item  There exist a sequence of non-overlapping finite time-intervals $[t_k,t_k+l_k]$  in which $\bigcup_{t=t_k}^{t_k+l_k}\mc{G}(t)$ 
		includes an undirected spanning  tree (uniform-connectivity). 
	\end{itemize}
\end{ass}

\ab{Next, we restate some nonsmooth set-value analysis from \cite{cortes2008discontinuous}. 
For a nonsmooth function $h:\mathbb{R}^m \rightarrow \mathbb{R}$, define its \textit{generalized gradient} as \cite{cortes2008discontinuous},}
\ab{\begin{align}
    \partial h(\mb{x})= \mathrm{co}\{\lim \nabla h(\mb{x}_i): \mb{x}_i \rightarrow \mb{x}, \mb{x}_i \notin \Omega_h \cup S \}
\end{align}
where $\mathrm{co}$ denotes convex hull, $S \subset \mathbb{R}^m$ is any set of zero Lebesgue measure, and $\Omega_h \in \mathbb{R}^m$ is the set of points in which $h$ is non-differentiable. If $h$ is \textit{locally Lipschitz} at $x$, then $\partial h(\mb{x})$ is nonempty, compact, and convex, and the set-valued map $\partial h:\mathbb{R}^m\rightarrow \mc{B}\{\mathbb{R}\}$ (with $\mc{B}\{\mathbb{R}\}$ denoting the collection of all subsets of $\mathbb{R}$), $\mb{x} \mapsto \partial h(\mb{x})$, is
upper semi-continuous and locally bounded. Then, its \textit{set-valued Lie-derivative} $\mc{L}_\mc{H} h : \mathbb{R}^m  \rightarrow \mathbb{R}$ 
with respect to system dynamics $\dot{\mb{x}} \in \partial \mc{H}(\mb{x})$ (with a unique solution) at $\mb{x}$ is,
\begin{align}
     \mc{L}_\mc{H} h =  \{\eta \in \mathbb{R}| \exists \nu \in \mc{H}(\mb{x})~s.t.~ \zeta^\top \nu = \eta,~\forall \zeta \in \partial h(\mb{x})\}
\end{align}
These are used for nonsmooth Lyapunov analysis in Section~\ref{sec_conv}.
} \vspace{-0.8cm}

\subsection{Preliminary Results on Convex Optimization}
Following the Karush-Kuhn-Tucker (KKT) condition  and Lagrange multipliers method, 
optimal solution to problem \eqref{eq_dra}  satisfies the \textit{feasibility condition} as described below.
\begin{defn} (\textbf{Feasibility Condition})
	Define $\mc{S}_\mb{b} = \{\mb{X} \in \mathbb{R}^{d \times n}|\mb{X}\mb{a} = \mb{b}\}$ and $\mb{X} \in \mc{S}_\mb{b}$ as the feasible set and value. \label{def_feas}
\end{defn}
\ab{Note that problem \eqref{eq_dra} differs from \textit{unconstrained} distributed optimization \cite{xin2019distributed,garg2021} due to \textit{feasibility} constraint $\mb{X}\mb{a} = \mb{b}$ which is of dimension $n-1$. Some works consider \textit{inequality constraints} $\mb{X}\mb{a} \leq \mb{b}$ \cite{turan2020resilient,nesterov2018dual,feijer2010stability}, which represents a half-space of dimension $n$, with example application in network utility maximization, saying that the weighted sum of utilities $\mb{X}\mb{a}$ should not \textit{exceed} certain value $\mb{b}$. These problems may encounter many such relaxed inequality constraints. In contrast, having one \textit{equality} constraint, e.g., in EDP, the weighted sum of generated power $\mb{X}\mb{a}$ should \textit{exactly} meet the load demand constraint $\mb{b}$ at all times, i.e., $\mb{X}\mb{a}=\mb{b}$  \cite{cherukuri2015tcns,yu2017distributed,molzahn2017survey,dominguez2012decentralized,kar2012distributed}. Having $m>1$ equality constraints, it can be algebraically reduced to a cost optimization of $n-m+1$ states subject to one feasibility constraint of dimension $n-m$, where the other $m-1$ states are dependent variables. For comparison of different constraints and solutions see \cite[Table~I]{mrd_2020fast}. 
\begin{defn}
Given a convex function $h(\mb{X}): \mathbb{R}^{d \times n} \rightarrow \mathbb{R}$,  the level set $L_\gamma(h)$ for a given $\gamma \in \mathbb{R}$ is  the set $L_\gamma(h)= \{\mb{X} \in \mathbb{R}^{d \times n}|h(\mb{X}) \leq \gamma\}$.
It is known that for a \textit{strictly} convex $h(\mb{X})$, all its level sets $L_\gamma(h)$ are also strictly convex, closed, and compact for all scalars $\gamma$.
\end{defn}}
\begin{lem}\label{lem_optimal_solution}
Problem~\eqref{eq_dra} under Assumption~\ref{ass_strict} has a unique optimal feasible solution $\mb{X}^* \in \mc{S}_\mb{b}$ as $\nabla \widetilde{F}(\mb{X}^*) = \pmb{\varphi}^* \otimes \mb{a}^\top$, with $\pmb{\varphi}^* \in \mathbb{R}^d$, $\widetilde{F}(\mb{X}) = \sum_{i=1}^{n}\widetilde{f}_i(\mb{x}_i)$, $\nabla \widetilde{F}(\mb{X}^*) = [\nabla \widetilde{f}_1(\mb{x}^*_1),\dots,\nabla \widetilde{f}_n(\mb{x}^*_n)]$ as the gradient (with respect to $\mb{X}$) of the function $\widetilde{F}$ at $\mb{X}^*$, and $\otimes$ as the Kronecker product.
\end{lem}
\begin{proof}
	The proof follows \cite{boyd2006optimal} by using KKT method. 
\end{proof}
In the following, we analyze  feasible solution set using the notion of \textit{level sets}. \ab{For two distinct $\mb{X}$ and $\mb{Y}$ with $h(\mb{X}) > h(\mb{Y})$ on two level sets $L_{h(\mb{X})}$ and $L_{h(\mb{Y})}$,  $\mb{e}_p^\top(h(\mb{Y}) - h(\mb{X}))\mb{e}_p> \mb{e}_p^\top \nabla h(\mb{X})(\mb{Y}-\mb{X})^\top \mb{e}_p$ and $\mb{e}_p^\top (h(\mb{X}) - h(\mb{Y})) \mb{e}_p> \mb{e}_p^\top \nabla h(\mb{Y})(\mb{X}-\mb{Y})^\top \mb{e}_p$; adding the two,
\begin{equation} \label{eq_level}
\mb{e}_p^\top (\nabla h(\mb{Y})-\nabla h(\mb{X}))(\mb{Y}-\mb{X})^\top \mb{e}_p > \mb{0},~p\in\{1,\dots,d\}
\end{equation}}
with $\mb{e}_p$ as the unit vector of the $p$'s coordinate.

\begin{lem}\label{lem_unique_feasible}
	 For every feasible set $\mc{S}_\mb{b}$ there exists only one unique point $\mb{X}^* \in \mc{S}_\mb{b}$ (under Assumption~\ref{ass_strict}) such that $\nabla\widetilde{F}(\mb{X}^*) = \Lambda \otimes \mb{a}^\top$ with $\Lambda \in \mathbb{R}^d$.
\end{lem}
\begin{proof}
	From strict convexity of $\widetilde{F}(\mb{X})$ (Assumption~\ref{ass_strict}), only one of its strict convex level sets, say $L_\gamma(\widetilde{F})$, touches the constraint facet $\mc{S}_\mb{b}$ only at a single point, say $\mb{X}^*$. 
	Clearly, the gradient $\nabla \widetilde{F}(\mb{X}^*)$ is orthogonal to  $\mc{S}_\mb{b}$, and 
	$\frac{\nabla \widetilde{f}_i(\mb{x}^*_i)}{a_i}=\frac{\nabla \widetilde{f}_j(\mb{x}^*_j)}{a_j} = \Lambda$ for all $i$.
    By contradiction, consider two points $\mb{X}^{*1},\mb{X}^{*2} \in \mc{S}_\mb{b}$ for which $\nabla \widetilde{F}({\mb{X}^{*1}}) = \Lambda_1 \otimes \mb{a}^\top$, $\nabla \widetilde{F}(\mb{X}^{*2}) = \Lambda_2 \otimes \mb{a}^\top$ (two possible optimum), implying that either (i) one level set $L_\gamma(\widetilde{F})$, $\gamma = \widetilde{F}(\mb{X}^{*1}) = \widetilde{F}(\mb{X}^{*2})$ is adjacent to the affine constraint $\mc{S}_\mb{b}$ at both $\mb{X}^{*1},\mb{X}^{*2}$, or (ii) there are two level sets $L_{\widetilde{F}(\mb{X}^{*1})}$ and $L_{\widetilde{F}(\mb{X}^{*2})}$, touching the affine set $\mc{S}_\mb{b}$ at $\mb{X}^{*1}$ and $\mb{X}^{*2}$ respectively, \ab{and thus, at both points $\nabla \widetilde{F}({\mb{X}^{*2}})$ and $\nabla \widetilde{F}({\mb{X}^{*2}})$ need to be orthogonal to $(\mb{X}^{*1}-\mb{X}^{*2})$ in $\mc{S}_\mb{b}$. Since $\mc{S}_\mb{b}$ forms a linear facet, the former case contradicts the strict convexity of the level sets. In the latter case,
    \begin{align} \label{eq_proof1}
    \mb{e}_p^\top (\nabla (\widetilde{F}({\mb{X}^{*1}})-\nabla \widetilde{F}({\mb{X}^{*2}}))(\mb{X}^{*1}-\mb{X}^{*2})^\top \mb{e}_p = 0,\forall p
    \end{align}
    which contradicts \eqref{eq_level}. This proves the lemma. }
\end{proof}
This proof analysis is further recalled in the next sections.

\section{The Proposed $1$st-Order Nonlinear Dynamics} \label{sec_dynamic}
We propose a $1$st-order protocol $\mc{F}: \mathbb{R}^{d\times n} \rightarrow \mathbb{R}^d$ coupling the agents' dynamics to solve problem \eqref{eq_dra}, while addressing model nonlinearities and satisfying  \textit{feasibility at all times}, 
\begin{align}
\dot{\mb{x}}_i = -\frac{1}{a_i}\sum_{j \in \mc{N}_i} W_{ij} g\Big(\frac{\nabla \widetilde{f}_i(\mb{x}_i)}{a_i} -  \frac{\nabla \widetilde{f}_j(\mb{x}_j)}{a_j}\Big) :~ \mc{F}_i(\mb{x}_i),
\label{eq_sol}
\end{align}
with  $W_{ij}$ as the weight of the link between agents $i$ and  $j$ and $\nabla \widetilde{f}_i(\mb{x}_i)$ as the gradient of  (time-invariant part of) the local objective $\widetilde{f}_i$ with respect to $\mb{x}_i$ and $g$ defines the nonlinearity to be explained later. \ab{ Following Assumption~\ref{ass_strict}, given a state point $\mb{X}_0$, the level set $L_{\widetilde{F}(\mb{X}_0)}$
is closed, convex, and compact.  Then, the solution set $L_{\widetilde{F}(\mb{X}_0)} \cap \mc{S}_b$ under \eqref{eq_sol}  is  closed and bounded.  Indeed \eqref{eq_sol} represents a \textit{differential inclusion} due to discontinuity of RHS of \eqref{eq_sol}  \cite{cortes2008discontinuous}, where for the sake of notation simplicity ''$=$'' is used instead of ''$\in$''. From \cite{cortes2008discontinuous}, it is straightforward to see that  the trajectory $\mc{F}$ is locally bounded, upper semi-continuous, with non-empty, compact, and convex values, and thus, from  
\cite[Proposition~S2]{cortes2008discontinuous} and  similar to \cite{garg2021,parsegov2013fixed}, the solution under \eqref{eq_sol} for initial condition $\mb{X}_0 \in \mb{S}_b$ exists and is unique. }
Recall that the time-varying and time-invariant parts of the local objectives are decoupled. Dynamics \eqref{eq_sol} represents a  \textit{$1$st-order weighted gradient tracking}, with no use of the Hessian matrix, Thus,  function $\widetilde{f}_i(\cdot)$ is not needed to be twice-differentiable (in contrast to $2$nd-order dynamics, e.g., in \cite{deng2019distributed}). This allows to incorporate smooth \textit{penalty functions} to address the \textit{box constraints}.  
In case of communication network among agents, \textit{periodic} communication with sufficiently small period $\tau$ is considered, see \cite{KIA2015254} for details.
The state of every agent $i$ evolves under influence of its direct neighbors $j \in \mc{N}_i$ weighted by $W_{ji}$, e.g, via information sharing networks \cite{KIA2015254} where every agent $i$  shares its local gradients $\nabla \widetilde{f}_i(\mb{x}_i)$ along with the weight $W_{ji}$. Therefore, the proposed resource allocation dynamics \eqref{eq_sol} is only based on local information-update, and is \textit{distributed} over the multi-agent network. 

\begin{ass} \label{ass_gW} (\textbf{Strongly sign-preserving nonlinearity})
In dynamics \eqref{eq_sol}, $g: \mathbb{R}^d \rightarrow  \mathbb{R}^d$ is a  nonlinear odd mapping such that $g(\mb{x}) = - g(-\mb{x})$, $g(\mb{x})\succ \mb{0}$ for $\mb{x}\succ \mb{0}$, $g(\mb{0}) = \mb{0}$, and $g(\mb{x})\prec\mb{0}$ for $\mb{x}\prec \mb{0}$. Further, $\nabla g(\mb{0})\neq \mb{0}$.  
\end{ass}
 
Some causes of such practical nonlinearities as function $g(\cdot)$ in \eqref{eq_sol}, e.g., physics-based nonlinearities, are given next. 

\textbf{Application 1:} Function $g(\cdot)$ can be adopted from  
finite-time  and fixed-time literature \cite{taes2020finite,parsegov2013fixed,mrd_2020fast} as $\mbox{sgn}^\mu(\mb{x})=\mb{x}\|\mb{x}\|^{\mu-1}$,
where $\|\cdot\|$ denotes the Euclidean norm and $\mu\geq 0$. In general, system dynamics as $\dot{\mb{x}}_i = -\sum_{j =1}^n W_{ij}(\mbox{sgn}^{\mu_1}(\mb{x}_i-\mb{x}_j)+\mbox{sgn}^{\mu_2}(\mb{x}_i-\mb{x}_j)) $ converge in finite/fixed-time \cite{parsegov2013fixed}, motivating fast-convergent allocation dynamics \cite{mrd_2020fast} as, 
\begin{align}
	\dot{\mb{x}}_i = &-\sum_{j \in \mc{N}_i} W_{ij}( \mbox{sgn}^{\mu_1}(\mb{z})+ \mbox{sgn}^{\mu_2}(\mb{z})), 
	\label{eq_sol_fixed}
\end{align}
with $\mb{z}=\frac{\nabla \widetilde{f}_i(\mb{x}_i)}{a_i} -  \frac{\nabla \widetilde{f}_j(\mb{x}_j)}{a_j}$,   $0<\mu_1<1$, and $0<\mu_2<1$ (finite-time case) or $1<\mu_2$ (fixed-time case).

\textbf{Application 2:} 
Quantized allocation by choosing  $g(\cdot)$  as,
\begin{align}
	g_{l}(\mb{z}) = \mbox{sgn}(\mb{z}) \exp(g_{u}(\log(|\mb{z}|))), \label{eq_quan_log}
\end{align}
where $g_{u}(\mb{z}) = \delta \left[ \frac{\mb{z}}{\delta}\right]$  represents the \textit{uniform} quantizer with  $\left[ \cdot\right]$ as rounding operation to the nearest integer \cite{wei2018nonlinear,frasca_quntized,guo2013consensus}, $\mbox{sgn}(\cdot)$ follows $\mbox{sgn}^\mu(\cdot)$ with $\mu =0$,  $\delta$ is the quantization level, and function $g_l$ denotes \textit{logarithmic} quantizer. 

\textbf{Application 3:} 
Sign-preserving nonlinear dynamics \cite{wei2017consensus,stankovic2020nonlinear} robust to  \textit{impulsive noise} can be achieved via $g_{p}(\mb{z}) = -\frac{d(\log p(\mb{z}))}{d \mb{z}}$ with $p$  as the noise density.
For example, for $p$ following  \textit{approximately uniform}  $\mc{P}_1$ or \textit{Laplace}  class $\mc{P}_2$ \cite{stankovic2020nonlinear},
\begin{align} \label{eq_robust_uni}
p \in \mc{P}_1 &: ~	g_p(\mb{z}) = 
\begin{cases}
	\frac{1 -\epsilon}{\epsilon d}\mbox{sgn}(\mb{z}) & |\mb{z}| > d\\
	0 & |\mb{z}| \leq d
\end{cases} 
\\ \label{eq_robust_sgn}
p \in \mc{P}_2 &: ~ g_p(\mb{z}) = 2 \epsilon \mbox{sgn}(\mb{z}),
\end{align}
with $0<\epsilon<1 $, $ d>0$. 

\textbf{Application 4:} Saturation nonlinearities  \cite{liu2019global,yi2019distributed} (or \textit{clipping}) are due to limited actuation range for which the saturation level may affect the stability, convergence, and general behavior of the system.
For a given saturation level $\kappa>0$, 
\begin{align} \label{eq_sat}
	g_\kappa(\mb{z}) = 
	\begin{cases}
		\kappa\mbox{sgn}(\mb{z}) & |\mb{z}| > \kappa\\
		\mb{z} & |\mb{z}| \leq \kappa
	\end{cases}
\end{align}	 	
\ab{
\begin{rem}
Recall that Eq. \eqref{eq_sol} represents \textit{Laplacian-gradient-type dynamics} (see \cite{cherukuri2015tcns} for details) which can ensure \textit{feasibility at all times under various nonlinearities of $g(\cdot)$} in contrast to the Lagrangian-type methods \cite{li2017distributed,dominguez2012decentralized,kar2012distributed,turan2020resilient,nesterov2018dual,feijer2010stability}. 
If the actuator is not subject to nonlinearities, one may select a linear function for $g(\cdot)$, i.e., $g(\mb{z}) = \mb{z}$ and utilize linear methods \cite{gharesifard2013distributed,cherukuri2015tcns,shames2011accelerated}. However, our focus is to provide a more general solution method that is applicable also to agents with nonlinearities (inherent or by design).
For example, the generators are known to be physically constrained with RRLs which is a determining factor on the stability of the grid \cite{hiskens_agc}. Linear methods cannot consider RRLs and may result in solutions with a high rate of change in power generation $\dot{\mb{x}}_i$, which cannot be followed in reality and may result in infeasibility or sub-optimality. However, such limits can be satisfied considering $g(\cdot)$ as in \eqref{eq_sat} where the limits can be tuned by  $\kappa$. 
\end{rem}  }

\section{Analysis of Convergence} \label{sec_conv}
In this section, combining convex analysis from Lemma~\ref{lem_optimal_solution}-\ref{lem_unique_feasible} with Lyapunov theory, we prove the convergence of the general protocol~\eqref{eq_sol} to the optimal value of problem \eqref{eq_dra} subject to the constraint on the weighted-sum of resources. The proof is, in general, irrespective of the nonlinearity types, i.e., holds for any nonlinearity satisfying Assumption~\ref{ass_gW}, including \eqref{eq_sol_fixed}-\eqref{eq_sat}.
\begin{lem} \label{lem_feasible_intime}
	(\textbf{Anytime Feasibility})
	Suppose  Assumption~\ref{ass_gW} holds. The states of the agents under dynamics~\eqref{eq_sol}  remain feasible, i.e.,  if $\mb{X}_0 \in \mc{S}_\mb{b}$, then $\mb{X}(t) \in \mc{S}_\mb{b}$ for $t>0$.	
\end{lem}
\begin{proof}
	Having $\mb{X}_0 \in \mc{S}_\mb{b}$ implies that $\mb{X}_0\mb{a}=\mb{b}$. For the general state dynamics \eqref{eq_sol},
	\small
	\begin{align} \label{eq_proof_feas}
	\frac{d}{dt}(\mb{X}\mb{a})=\sum_{i=1}^n \dot{\mb{x}}_ia_i = -\sum_{i=1}^n\sum_{j \in \mc{N}_i} W_{ij} g\Big(\frac{\nabla \widetilde{f}_i(\mb{x}_i)}{a_i} -  \frac{\nabla \widetilde{f}_j(\mb{x}_j)}{a_j}\Big).
	\end{align} \normalsize
	From Assumptions~\ref{ass_G} and~\ref{ass_gW},  $W_{ij}=W_{ji}$ and $g(-\mb{x})=-g(\mb{x})$.
	Therefore, the summation in \eqref{eq_proof_feas} is equal to zero,  ${\frac{d}{dt}(\mb{X}\mb{a})=\mb{0}}$, and $\mb{X}\mb{a}$ is time-invariant under dynamics \eqref{eq_sol}. Thus, having feasible initial states $\mb{X}_0\mb{a}=\mb{b}$, then ${\mb{X}(t)\mb{a}=\mb{b}}$ remains feasible over time, i.e. $\mb{X}(t) \in \mc{S}_{\mb{b}}$ for all $t>0$. 
\end{proof}

\ab{The above proves \textit{anytime feasibility}, i.e., nonlinear dynamic \eqref{eq_sol} remains feasible \textit{at all times}, which is privileged over consensus-based solutions \cite{dominguez2012decentralized,kar2012distributed,li2017distributed}. For AGC subject to RRL, $\dot{\mb{x}}_i$ (and thus $g(\cdot)$) needs to be further of limited range.
Further, Lemma~\ref{lem_feasible_intime} shows that $\mc{S}_b$ is \textit{positively invariant} under the nonlinear dynamics  \eqref{eq_sol}.}

\begin{thm} \label{thm_tree}
	(\textbf{Equilibrium-Uniqueness})
	Under Assumptions~\ref{ass_G} and \ref{ass_gW}, the equilibrium point $\mb{X}^*$ of the solution dynamics \eqref{eq_sol} is only in the form $\nabla \widetilde{F}(\mb{X}^*) = \Lambda \otimes \mb{a}^\top$ with $\Lambda \in \mathbb{R}^d$, and coincides with the unique optimal point of \eqref{eq_dra}.
\end{thm}
\begin{proof}
	From dynamics~\eqref{eq_sol}, $\dot{\mb{x}}^*_i = \mb{0},\forall i$ for $\mb{X}^*$ satisfying $\nabla \widetilde{F}(\mb{X}^*) = \Lambda \otimes \mb{a}^\top$, and such point $\mb{X}^*$ is clearly an equilibrium of~\eqref{eq_sol}. We prove that there is no other equilibrium with $\nabla \widetilde{F}(\mb{X}^*) \neq \Lambda \otimes \mb{a}^\top$ by contradiction. Assume $\widehat{\mb{X}}$ as the equilibrium of~\eqref{eq_sol} such that $  \frac{\nabla \widetilde{f}_i(\widehat{\mb{x}}_i)}{a_i}\neq\frac{\nabla \widetilde{f}_j(\widehat{\mb{x}}_j)}{a_j}$
	for at least two agents $i,j$. Let $\nabla \widetilde{F}({\widehat{\mb{X}}}) = (\widehat{\Lambda}_1,\dots,\widehat{\Lambda}_n)$. Consider two agents $\alpha = \argmax_{q\in \{1,\dots,n\}}  \widehat{\Lambda}_{q,p}$ and $\beta = \argmin_{q \in \{1,\dots,n\}}  \widehat{\Lambda}_{q,p}$ for any entry $p \in  \{1,\dots,d\}$.	
	Following the Assumption~\ref{ass_G}, the existence of an (undirected) spanning tree in the union network  $\bigcup_{t=t_k}^{t_k+l_k}\mc{G}(t)$ implies that there is a mutual path between nodes (agents) $\alpha$ and $\beta$. In this path, there exists at least two agents $\overline{\alpha}$ and $\overline{\beta}$ for which
        $\widehat{\Lambda}_{\overline{\alpha},p} \geq \widehat{\Lambda}_{\mc{N}_{\overline{\alpha}},p},~ \widehat{\Lambda}_{\overline{\beta},p}\leq \widehat{\Lambda}_{\mc{N}_{\overline{\beta}},p}$
    with $\mc{N}_{\overline{\alpha}}$ and $\mc{N}_{\overline{\beta}}$ as the neighbors of $\overline{\alpha}$ and $\overline{\beta}$, respectively. The strict inequality holds for at least one neighboring node in $\mc{N}_{\overline{\alpha}}$ and $\mc{N}_{\overline{\beta}}$. From Assumption~\ref{ass_G} and \ref{ass_gW}, in a sub-domain of $[t_k,t_k+l_k]$, we have $\dot{\widehat{\mb{x}}}_{\overline{\alpha},p} <0$ and $\dot{\widehat{\mb{x}}}_{\overline{\beta},p} > 0$. Therefore, $\dot{\widehat{\mb{X}}} \neq \mb{0}$ which contradicts the assumption that $\widehat{\mb{X}}$ is the equilibrium of~\eqref{eq_sol}. Recall that, from Lemma~\ref{lem_unique_feasible}, this point coincides with the optimal solution of \eqref{eq_dra}, as   for every feasible initialization in $\mc{S}_b$ there is only one such point $\mb{X}^*$ satisfying $\nabla\widetilde{F}(\mb{X}^*) = \Lambda \otimes \mb{a}^\top$. 
     This completes the proof.
\end{proof}


The above lemma paves the way for convergence analysis via the Lyapunov stability theorem, as it shows that the dynamics \eqref{eq_sol} \textit{has a unique equilibrium for any feasible initial condition.}

\begin{lem} \cite[Lemma~3]{mrd_2020fast} \label{lem_sum}
	Let nonlinearity $g(\cdot) $ and matrix $W$ satisfy Assumptions~\ref{ass_G} and \ref{ass_gW}. Then, for $\pmb{\psi} \in \mathbb{R}^d$ we have, \vspace{-0.5cm}
	
	\small
	\begin{align} \nonumber
	\sum_{i =1}^n \pmb{\psi}_i^\top\sum_{j =1}^nW_{ij}g(\pmb{\psi}_j-\pmb{\psi}_i)= \sum_{i,j =1}^n \frac{W_{ij}}{2} (\pmb{\psi}_j-\pmb{\psi}_i)^\top g(\pmb{\psi}_j-\pmb{\psi}_i).
	\end{align} \normalsize
\end{lem}
Following the convex analysis in Lemmas~\ref{lem_unique_feasible}-\ref{lem_feasible_intime}, and Theorem~\ref{thm_tree} along with Lemma~\ref{lem_sum}, we provide our main theorem next.

\begin{thm} \label{thm_converg}
	(\textbf{Convergence})
	Suppose Assumptions~\ref{ass_strict}-\ref{ass_gW} hold. Then, initializing by $\mb{X}_0 \in \mc{S}_\mb{b}$, the proposed dynamics \eqref{eq_sol} solves the  network resource allocation problem \eqref{eq_dra}.
\end{thm}

\begin{proof}
	Following Lemmas~\ref{lem_unique_feasible},~\ref{lem_feasible_intime}, and Theorem~\ref{thm_tree} and initializing from $\mb{X}_0 \in \mc{S}_\mb{b}$ for any $b \in \mathbb{R}^d$, there is a unique feasible equilibrium $\mb{X}^*$ for solution dynamics~\eqref{eq_sol} in the form $\nabla \widetilde{F}(\mb{X}^*) = \pmb{\varphi}^* \otimes \mb{a}^\top$. Define the nonsmooth residual function $\overline{F}(\mb{X})=F(\mb{X},t)-F(\mb{X}^*,t)$.
	Clearly, $\overline{F}(\mb{X})=\sum_{i=1}^n (\widetilde{f}_i(\mb{x}_i)-\widetilde{f}_i(\mb{x}^*_i))>0$ is \textit{purely a function of $\mb{X}$}, with  $\mb{X}^*$ as its unique equilibrium. \ab{For this continuous (but nonsmooth) regular and locally Lipschitz Lyapunov function $\overline{F}(\mb{X})$,  its generalized derivative $t \mapsto \overline{F}(\mb{x}(t))$, for $\mb{x}$ as the solution to \eqref{eq_sol}, satisfies $\partial_t \overline{F}(\mb{X}(t)) \in \mc{L}_\mc{F} \overline{F}(\mb{X}(t))$, see  \cite[Proposition~10]{cortes2008discontinuous}. Then (dropping $t$ for notation simplicity),} 
	
	\ab{
	\footnotesize	\begin{align}
	\partial_t \overline{F} 
	= \nabla F^\top \dot{\mb{X}} = \sum_{i =1}^n -\frac{\nabla \widetilde{f}_i(\mb{x}_i)}{a_i}^\top\sum_{j \in \mc{N}_i} W_{ij} g\Big(\frac{\nabla \widetilde{f}_i(\mb{x}_i)}{a_i} -  \frac{\nabla \widetilde{f}_j(\mb{x}_j)}{a_j}\Big). \nonumber
	\end{align} \normalsize
	Following Lemma~\ref{lem_sum},
	\footnotesize
	\begin{align} \label{eq_proof_decay}
	\partial_t \overline{F} =  -\sum_{i,j =1}^n \frac{W_{ij}}{2}\Big(\frac{\nabla \widetilde{f}_i(\mb{x}_i)}{a_i} -  \frac{\nabla \widetilde{f}_j(\mb{x}_j)}{a_j}\Big)^\top g\Big(\frac{\nabla \widetilde{f}_i(\mb{x}_i)}{a_i} -  \frac{\nabla \widetilde{f}_j(\mb{x}_j)}{a_j}\Big).
	\end{align}  \normalsize
	From Assumption~\ref{ass_gW}, $g(\mb{x})$ is odd and strongly sign-preserving, i.e., $\mb{x}^\top g(\mb{x})\geq0$.
	Therefore, $\partial_t \overline{F} \leq 0$ with the largest invariant set $\mc{I}$ contained in $ \{\mb{X} \in L_{\widetilde{F}(\mb{X}_0)} \cap \mc{S}_b| \mb{0} \in \mc{L}_\mc{F} \overline{F}(\mb{X})\}$, i.e., $\mc{I}$ includes the unique point $\mb{X}^*\in \mc{S}_b$ for which $\nabla \widetilde{F}(\mb{X}) \in \mbox{span}\{\mb{a}\}$ (or $\frac{\nabla \widetilde{f}_i(\mb{x}^*_i)}{a_i} =  \frac{\nabla \widetilde{f}_j(\mb{x}^*_j)}{a_j}=\pmb{\varphi}^*,~\forall i,j$) from Lemmas~\ref{lem_optimal_solution} and \ref{lem_unique_feasible}.
	 Using  LaSalle invariance principle for differential
    inclusions \cite[Theorem~2.1]{cherukuri2015tcns}, initializing by $\mb{X}_0 \in \mc{S}_b$, the trajectory set $\{L_{\widetilde{F}(\mb{X}_0))} \cap \mc{S}_b\}$ remains feasible and positively invariant under \eqref{eq_sol} (Lemma~\ref{lem_feasible_intime}), and converges to the largest invariant set $\mc{I} = \{\mb{X}^*\}$ including the  unique equilibrium of \eqref{eq_sol} 
    (as shown in Theorem~\ref{thm_tree}),    $\overline{F}$ is monotonically non-decreasing and radially unbounded, $\max \mc{L}_\mc{F} \overline{F}(\mb{X}(t)) <0$ for all $\mb{X \in \mc{S}_b \setminus \mc{I}}$, and thus, from \cite[Theorem~1]{cortes2008discontinuous}
    $\mb{X}^*$ is globally strongly asymptotically stable}. This proves that agents' states under dynamics \eqref{eq_sol} converge to  $\mb{X}^*$.
\end{proof}
\ab{The above proof holds for any $b$ value and any initialization state $\mb{X}_0 \in \mb{S}_b$, and the solution converges to $\mb{X}^*$ in Lemma~\ref{lem_optimal_solution}. } 
\begin{rem} 
 \ab{Following similar analysis as in \cite{cherukuri2015tcns}, 
 assuming $\exists u_{\min},K_{\min}$ such that $u_{\min} \leq \nabla^2 f_i(\mb{x}_i)$ (strongly convex cost with smooth gradient) and $ K_{\min}  \leq \frac{g(\mb{z})}{\mb{z}}$, Eq. \eqref{eq_proof_decay}  over a connected network $\mc{G}$ with $\lambda_2$ as its algebraic connectivity (Fiedler-value) and $\mb{a}=\mb{1}_n$ gives the decay rate of $\overline{F}$ as,
 \begin{align}
     \partial_t \overline{F} \leq -2 u_{\min} K_{\min} \lambda_2 \overline{F}
 \end{align}}
  \ab{For a disconnected network with at least one link $(i,j)$, the summation in \eqref{eq_proof_decay} is positive and $\partial_t \overline{F}$ is negative if  $\frac{\nabla \widetilde{f}_i(\mb{x}_i)}{a_i} \neq \frac{\nabla \widetilde{f}_j(\mb{x}_j)}{a_j}$.
  From Assumption~\ref{ass_G}, $\partial_t \overline{F}$ is  negative over sub-intervals of every time-interval $[t_k~t_k+l_k]$ (infinitely often) having $\frac{\nabla \widetilde{f}_i(\mb{x}_i)}{a_i} \neq \frac{\nabla \widetilde{f}_j(\mb{x}_j)}{a_j}$ for (at least) $2$ neighbors $i,j$ till reaching the optimizer $\mb{X}^*$ (for which $\frac{\nabla \widetilde{f}_i(\mb{x}^*_i)}{a_i} = \frac{\nabla \widetilde{f}_j(\mb{x}^*_j)}{a_j}~\forall i,j$). One may also consider discrete Lyapunov analysis and simply prove that $\overline{F}(\mb{X}(t_k+l_k))<\overline{F}(\mb{X}(t_k))$ for all $\mb{X}(t_k) \in \mb{\mc{S}_b \setminus \mc{I}}$.}
\end{rem}

\section{Simulation over Sparse Networks}
\label{sec_sim}
\begin{figure}[t]
	\centering
	\includegraphics[width=1.7in]{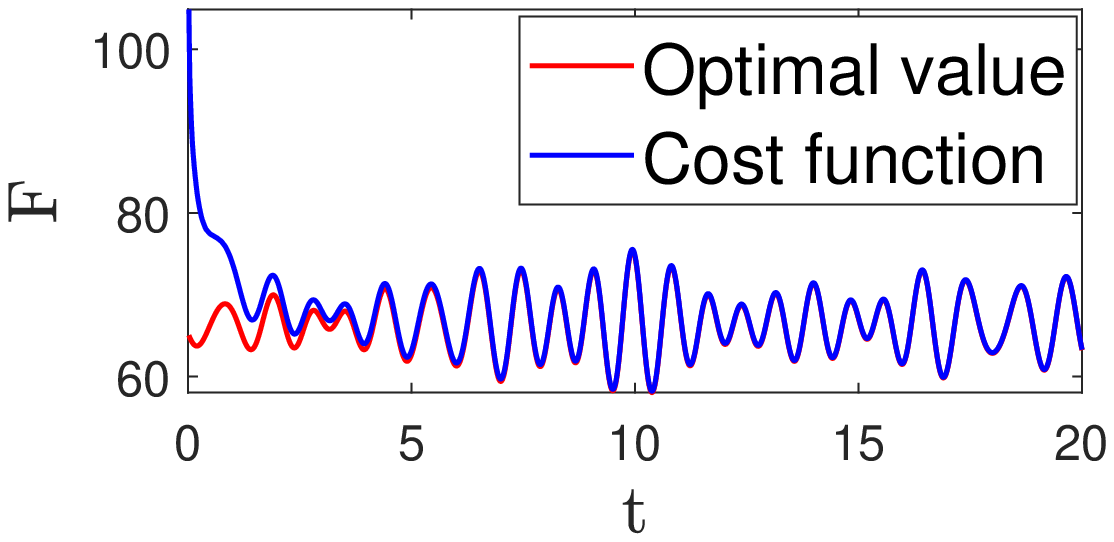}
	\includegraphics[width=1.7in]{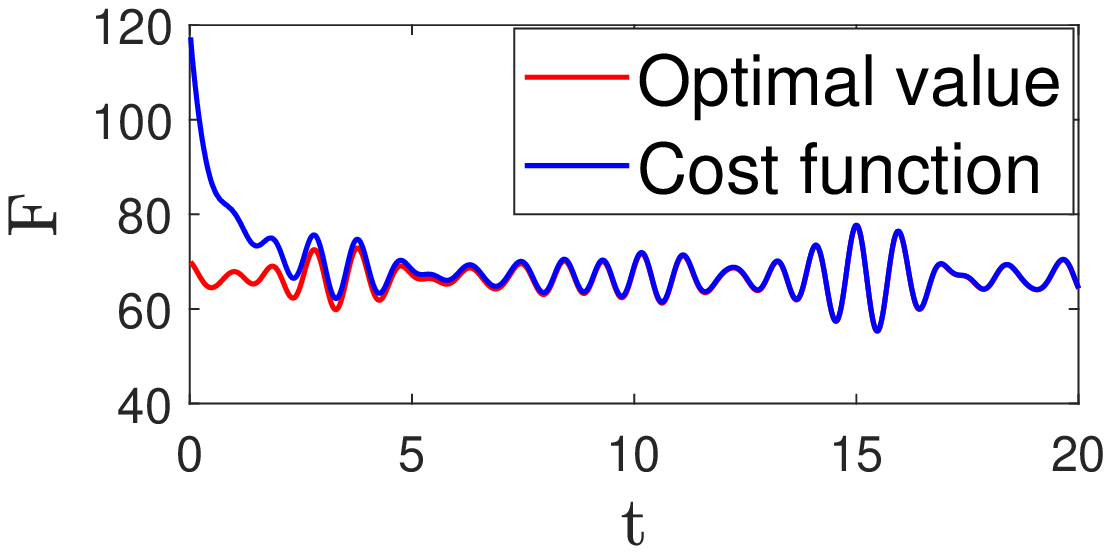}
	\includegraphics[width=1.7in]{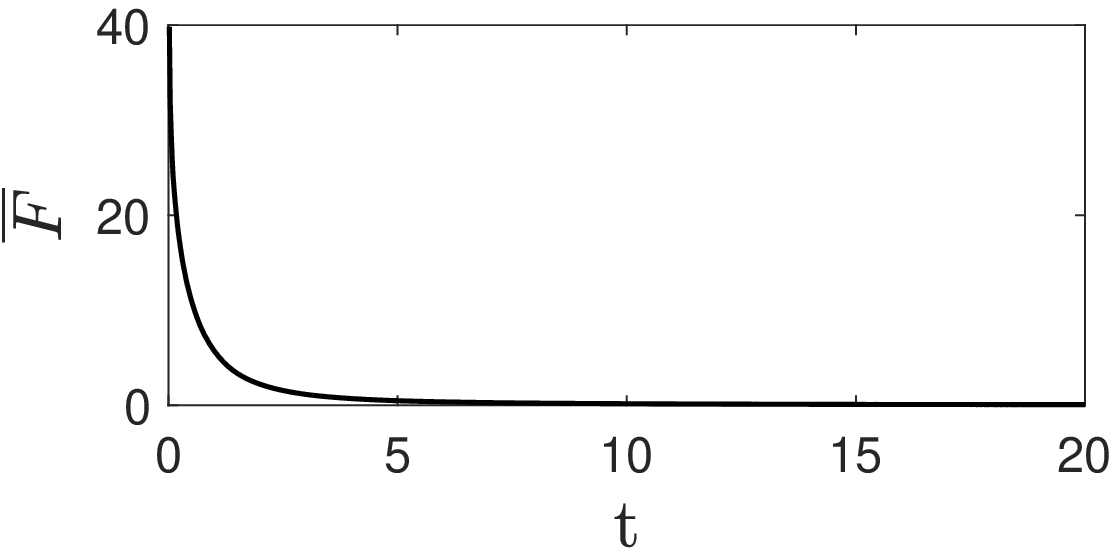}
	\includegraphics[width=1.7in]{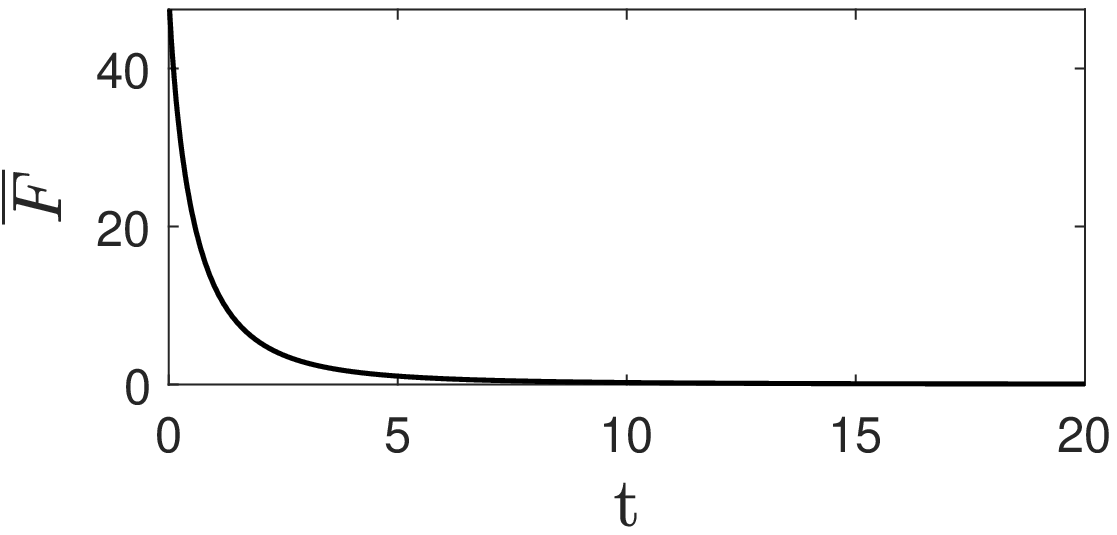}
	\caption{ The time-evolution of (Top) the cost function versus the time-varying optimal value and (Bottom) the associated Lyapunov function for quantized resource allocation over example switching networks for quantized (Left) and saturated (Right) actuation dynamics for resource allocation.}
	\label{fig_cost} 
\end{figure}
We simulate protocol~\eqref{eq_sol} for (i) quantized and (ii) saturated resource allocation over  
$4$ weakly-connected Erdos-R{\'e}nyi networks of $n=100$ agents
changing every $0.1$ second 
with switching command $s:\lceil 10t-4\lfloor 2.5t \rfloor\rceil$ satisfying Assumption~\ref{ass_G}. 
Consider strictly convex cost as \cite{boyd2006optimal},

\small\begin{align} \label{eq_F2}
\begin{cases}
\widetilde{f}_i(\mb{x}_i) &= \sum_{j=1}^4 \bar{a}_{i,j}(\mb{x}_{i,j}-\bar{c}_{i,j})^2 \\ & ~~+ \log(1+\exp(\bar{b}_{i,j}(\mb{x}_{i,j}-\bar{d}_{i,j}))), 
\\
\widehat{f}_i(t) &= \sum_{j=1}^4 \bar{e}_{i,j}\sin(\alpha_{i,j}t+\phi_{i,j})
\end{cases}
\end{align} \normalsize
with random parameters. Assume $ \mb{b}= 10\mb{1}_4$ and $a_i$ in $[0.1,1]$. To solve~\eqref{eq_dra}, we accommodate \eqref{eq_sol} for two cases:  (i) quantized actuation via the logarithmic quantizer \eqref{eq_quan_log} with $\delta=1$, and (ii) saturated actuation \eqref{eq_sat} with $\kappa=1$. The time-evolution of the cost \eqref{eq_F2} and the Lyapunov $\overline{F}(\mb{X})=F(\mb{X},t)-F^*(t)$ are shown in Fig.~\ref{fig_cost}.
As it is clear, the cost functions converge to the optimal (time-varying) values, with Lyapunov functions (residuals) decreasing in time.

\section{Application: Automatic Generation Control} \label{sec_edp}
\ab{The AGC adjusts the power generation based on predetermined reserve limits to compensate for any generation-load mismatch in a time scale of minutes. We assume that the generation-load mismatch is known (e.g. generator outage) and we aim to allocate that mismatch to the generators by minimizing their power deviation cost. 
Let $\mb{x}_i$ represent the power deviation for  generator $i$. The optimization problem finds the optimal mismatch allocation to $n$ generators while satisfying the reserve limits and is given by:}

\ab{
\begin{align} \label{eq_f_quad}
\min_\mb{X}  &\sum_{i=1}^n \gamma_i \mb{x}_i^2+ \beta_i \mb{x}_i + \alpha_i,\\ ~\mbox{s.t.}~&\sum_{i=1}^n \mb{x}_i = P_{mis}, ~~ -\underline{R}_i \leq \mb{x_i} \leq \overline{R}_i,~  i=1,..,n. \nonumber
\end{align}
The generation-load mismatch is $P_{mis}$ and the reserve limits for decreasing and increasing the power generation are  $\underline{R}$ and  $ \overline{R}$, respectively.
Mapping the problem to formulation \eqref{eq_dra},
$d=1$, $a_i=1$, ${\mb{b}=P_{mis}}$, $\underline{m}=\underline{R}_i$, $\overline{m}=\overline{R}_i$.
The example of Fig. \ref{fig_edp1} is derived using 
${n=10}$, $P_{mis}=800~MW$, $\underline{R}_i=50$, $\overline{R}_i=150$, and a set of realistic generator cost parameters.
The initial allocated power is  $\frac{D}{n}=80~MW$. 
We apply (i) the dynamics~\eqref{eq_sol_fixed} with $\mu_1=0.7$, $\mu_2=1.4$ and (ii)  the robustified  dynamics~\eqref{eq_sol} via $g_p(\cdot)$ in \eqref{eq_robust_sgn} with $\epsilon = 0.5$ to optimally allocate power over a cyclic communication network with random link weights. 
We compare the results with linear  \cite{cherukuri2015tcns,boyd2006optimal}, accelerated linear  \cite{shames2011accelerated}, finite-time  \cite{chen2016distributed},  and initialization-free \cite{yi2016initialization} protocols   in Fig.~\ref{fig_edp1}.
\begin{figure}[t]
	\centering
	\includegraphics[width=3.1in]{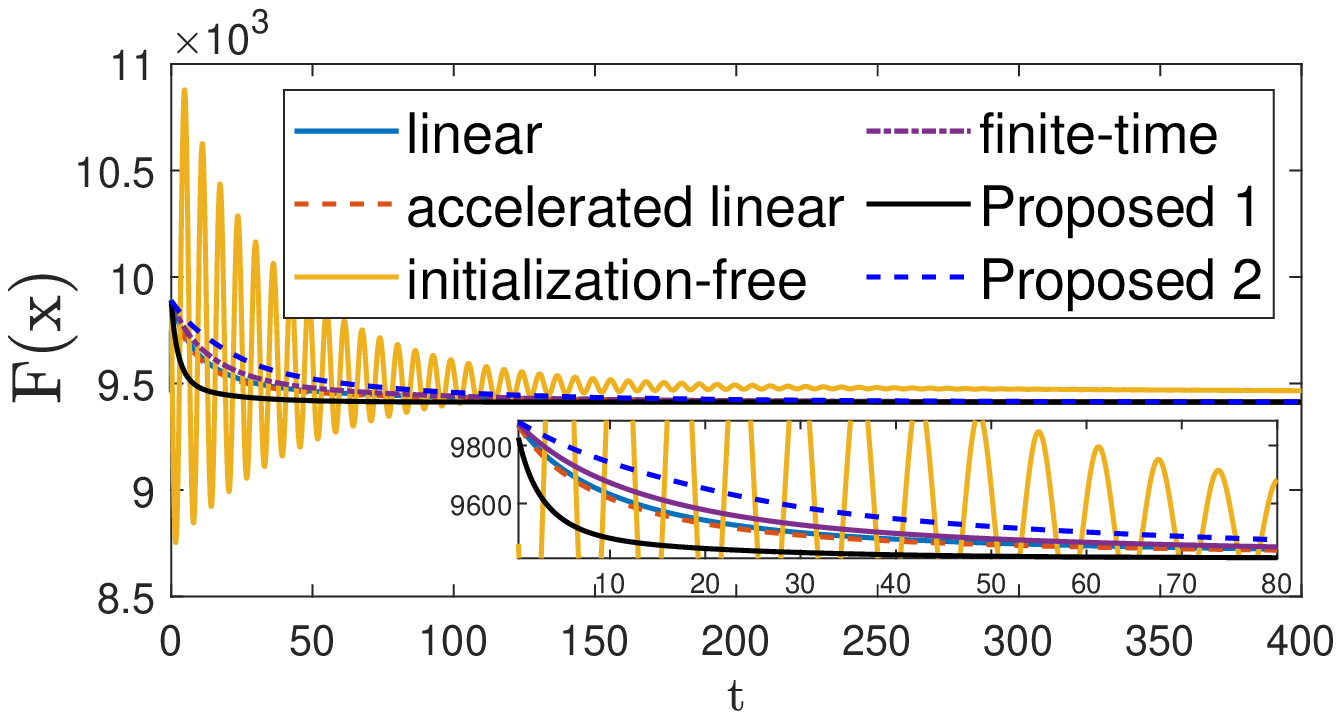}
	\includegraphics[width=3.1in]{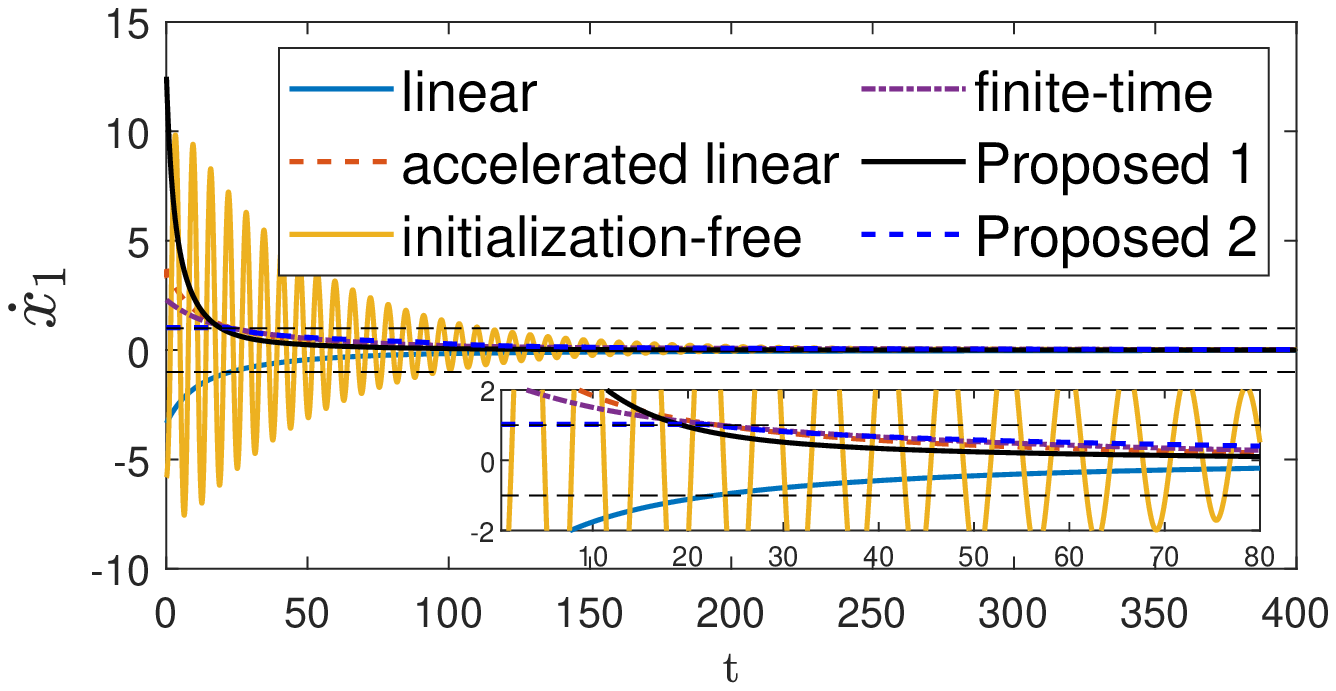}
	\caption{(Top) This figure compares the residual of the  dynamics~\eqref{eq_sol_fixed} (solid black) and robustified~\eqref{eq_sol} via \eqref{eq_robust_sgn} (dashed blue) with some recent literature. (Bottom) The power rates are compared at one generator. The horizontal dashed lines represent $\pm 1\frac{MW}{min}$ as the RRL. Only the robustified saturated dynamics (Proposed 2) met these limitations. 
	} \vspace{-1.3cm}
	\label{fig_edp1}
\end{figure}
From Fig.~\ref{fig_edp1}, considering RRL in the context of AGC, clearly the 
robustified dynamics converges with fixed-rate over time to keep the generation power within the ramp-limits (dashed blue), while other solutions impose a high rate of power generation that is impossible for the generators to follow. Such \textit{rate-constraints} cannot be easily addressed  via primal-dual \cite{turan2020resilient,nesterov2018dual,feijer2010stability} methods.
Note that in case there are no RRL requirements, the proposed fixed-time protocol (solid black) converges faster than the linear and other solutions. }
\vspace{-0.2cm}
\section{Conclusion} \label{sec_conclusion}
This paper proposes general nonlinear-constrained  solutions for resource allocation over uniformly-connected networks. 
The proposed solution can solve the allocation problem subject to nonlinearities in Applications (i)-(iv) in Section~\ref{sec_dynamic},  
their composition  mapping (as it is also odd and strongly sign-preserving), or any other nonlinearity satisfying Assumption~\ref{ass_gW}.

\bibliographystyle{IEEEbib}
\bibliography{bibliography}
\end{document}